\documentclass[a4paper,accepted=2020-10-21]{quantumarticle}
\pdfoutput=1
\usepackage[utf8]{inputenc}
\usepackage[T1]{fontenc}
\usepackage[british]{babel}
\usepackage[unicode,colorlinks]{hyperref}
\usepackage{graphicx}
\usepackage{amsmath,amssymb,amsthm}
\usepackage{xspace}
\usepackage{mathtools}
\usepackage{dsfont}
\usepackage{xcolor}
\usepackage{tikz}
\usepackage{soul} 

\DeclareMathOperator{\tr}{tr}

\let\originalleft\left
\let\originalright\right
\renewcommand{\left}{\mathopen{}\mathclose\bgroup\originalleft}
\renewcommand{\right}{\aftergroup\egroup\originalright}

\newcommand{\bra}[1]{\left\langle #1 \right|}
\newcommand{\ket}[1]{\left| #1 \right\rangle}

\newcommand{\ketbra}[2]{\left|#1\middle\rangle\middle\langle#2\right|}
\newcommand{\proj}[1]{\left|#1\middle\rangle\middle\langle#1\right|}

\newcommand{\abs}[1]{\left|#1\right|}
\newcommand{\mean}[1]{\left\langle#1\right\rangle}

\newcommand{\de}[1]{\left(#1\right)}
\newcommand{\De}[1]{\left[#1\right]}

\newcommand{\mathand}{\quad\text{and}\quad}

\newcommand{\ie}{i.e.\@\xspace}

\newtheorem{theorem}{Theorem}
\newtheorem*{theorem*}{Theorem}
\newtheorem{lemma}[theorem]{Lemma}

\mathtoolsset{centercolon}

\newtheorem{result}{Result}
   
\hypersetup{pdftitle={{Bell nonlocality with a single shot}}}

\begin{document}
\title{Bell nonlocality with a single shot}
\author{Mateus Araújo}\orcid{0000-0003-0155-354X}
\affiliation{Institute for Quantum Optics and Quantum Information (IQOQI),
Austrian Academy of Sciences, Boltzmanngasse 3, 1090 Vienna, Austria}
\author{Flavien Hirsch}\orcid{0000-0002-4145-7299}
\affiliation{Institute for Quantum Optics and Quantum Information (IQOQI),
Austrian Academy of Sciences, Boltzmanngasse 3, 1090 Vienna, Austria}
\author{Marco Túlio Quintino}\orcid{0000-0003-1332-3477}
\affiliation{Vienna Center for Quantum Science and Technology (VCQ), Faculty of Physics,
 University of Vienna, Boltzmanngasse 5, 1090 Vienna, Austria}
\affiliation{Institute for Quantum Optics and Quantum Information (IQOQI),
Austrian Academy of Sciences, Boltzmanngasse 3, 1090 Vienna, Austria}
\affiliation{Department of Physics, Graduate School of Science, The University of Tokyo, Hongo 7-3-1, Bunkyo-ku, Tokyo 113-0033, Japan}

\date{10th March 2025}

\begin{abstract}
In order to reject the local hidden variables hypothesis, the usefulness of a Bell inequality can be quantified by how small a $p$-value it will give for a physical experiment. Here we show that to obtain a small expected $p$-value it is sufficient to have a large gap between the local and Tsirelson bounds of the Bell inequality, when it is formulated as a nonlocal game. We develop an algorithm for transforming an arbitrary Bell inequality into an equivalent nonlocal game with the largest possible gap, and show its results for the CGLMP and $I_{nn22}$ inequalities.

We present explicit examples of Bell inequalities with gap arbitrarily close to one, and show that this makes it possible to reject local hidden variables with arbitrarily small $p$-value in a single shot, without needing to collect statistics. We also develop an algorithm for calculating local bounds of general Bell inequalities which is significantly faster than the naïve approach, which may be of independent interest.
\end{abstract}
\maketitle

\begin{figure}[htb]
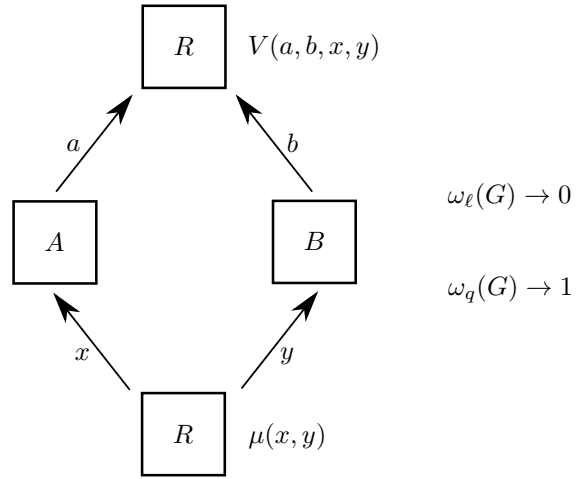
\caption{Schematic representation of a bipartite nonlocal game. A referee samples questions $x$ and $y$ with probability $\mu(x,y)$ and sends them to Alice and Bob. They send their answers $a$ and $b$ back to the referee, who accepts their answers with probability $V(a,b,x,y)$, in which case they win. If the maximal probability of winning the game with local hidden variables $\omega_\ell(G)$ is close to zero and the probability of winning it with the optimal quantum strategy $\omega_q(G)$ is close to one then this nonlocal game makes it possible to reject local hidden variables in a single round.}
\end{figure}

One of the most impactful developments of modern physics was the discovery of Bell's theorem \cite{bell1964}, that implies that the predictions of quantum mechanics cannot be reproduced by local and deterministic theories \cite{chsh1969} or, in a different formulation, that they cannot be reproduced by locally causal theories \cite{bell1975}. It forced us to rethink cherished notions of determinism and locality, and perhaps even the existence of many worlds \cite{deutsch2000,brown2016}.

The contradiction between quantum mechanics and local hidden variables (LHVs) provided by Bell's theorem is inherently probabilistic, as it is always possible, although unlikely, to obtain experimental results consistent with quantum mechanics merely by chance. This makes it a nontrivial problem to reject LHVs experimentally. Historically, the approach chosen was to estimate the value of the Bell expression by making a large number of measurements, calculate the variance of this estimator, and declare LHVs rejected if this estimate was some number of standard deviations above the local bound \cite{aspect1982,weihs1998,rowe2001}. This is not satisfactory, though, as it still leaves open the possibility that LHVs could reproduce these experimental results with high probability. In the recent loophole-free Bell tests \cite{hensen2015,giustina2015,shalm2015}, the approach taken was instead to calculate the probability that LHVs could reproduce the observed data, and declare them rejected if this $p$-value was below some threshold.

To obtain a stronger -- or easier -- rejection of LHVs, it is then interesting to search for the form of the Bell inequality that minimizes this $p$-value. The best way to address this question is to reformulate the Bell inequality as a nonlocal game \cite{cleve2004}, as the expected $p$-value decreases monotonically with the size of the gap of the nonlocal game, which is the difference between its Tsirelson bound and its local bound. We develop an algorithm to translate a Bell inequality into an equivalent nonlocal game with the largest possible gap, showing that finding such an optimal game reduces to solving a linear programming problem. We present analytical solutions for unique games and the CGLMP inequalities, and numerical solutions for the $I_{nn22}$ inequalities.

A crucial feature of nonlocal games is that they are played in rounds, and in each round the players either win or lose. This makes it possible to calculate the $p$-value of obtaining any number of victories in any number of rounds, and raises the natural question of what is the minimal number of rounds necessary to reject LHVs with a given $p$-value. Perhaps surprisingly, we show that it is possible to do so with a single measurement, for any chosen $p$-value, as had been speculated before in Refs.~\cite{barrett2002,gill2001}. To see that, note that the $p$-value of obtaining a single victory in a single round of the nonlocal game with LHVs is simply the local bound of game. If it is smaller than the desired $p$-value, obtaining such a victory is enough to reject LHVs. If the Tsirelson bound of the game is close to one, it is very likely that this will happen when playing the game with quantum mechanics.

We present two ways of constructing nonlocal games with the desired properties. The first is by using the parallel repetition technique, well-known in computer science, where we play $n$ instances of a nonlocal game in parallel, in a single round, instead of in $n$ consecutive rounds as it is usually done. As shown by Rao \cite{rao2011}, this will turn any nonlocal game with Tsirelson bound strictly larger than the local bound into a nonlocal game with local bound arbitrarily close to zero and Tsirelson bound arbitrarily close to one. The second construction uses the Khot-Vishnoi game; as shown by Kempe et al. \cite{kempe2007} there exists a choice of parameter for which its local bound is arbitrarily close to zero and its Tsirelson bound is arbitrarily close to one. In both cases a quantum state with unreasonably large dimension is required to obtain a single-shot rejection of LHVs.

This raises the question of what is the minimal quantum state dimension required to obtain a given gap. We show that the size of the gap is closely related to the so-called largest violation of a Bell inequality \cite{junge2010,junge2010b}, and this relationship allows us to derive upper bounds on the size of the gap as a function of the dimension. The gaps of the constructions from the previous paragraph are much smaller than the upper bounds we found, which raises the possibility of constructing a nonlocal game that makes possible a much easier single-shot rejection of LHVs.

The paper is organised as follows: Section \ref{sec:basic} introduces Bell inequalities and nonlocal games. Section \ref{sec:pvalue} discusses how to obtain the $p$-value for a given nonlocal game, and shows that nonlocal games with a large gap have small $p$-value. Section \ref{sec:optimal} presents the algorithm for transforming a Bell inequality into an equivalent nonlocal game with largest possible gap, and present its results for the CGLMP and $I_{nn22}$ inequalities. Section \ref{sec:singleshot} presents the nonlocal games that allow rejection of LHVs with a single shot. Section \ref{sec:asymptotics} presents the bounds on the size of the gap. Section \ref{sec:flavien} presents the algorithm for calculating local bounds.

\section{Bell inequalities and nonlocal games}\label{sec:basic}

There are two main approaches to the study of Bell nonlocality. In the physics literature it is common to use Bell inequalities. In a scenario where two non-communicating parties, Alice and Bob, produce outcomes $a$ and $b$ when given settings $x$ and $y$, a Bell inequality is the expression
\begin{equation}\label{eq:Bell}
\sum_{abxy}M^{ab}_{xy} p(ab|xy) \leq L
\end{equation} 
where $p(ab|xy)$ are conditional probabilities, $M^{ab}_{xy}$ are real coefficients and $L$ is the \textit{local bound}, which is the maximal value of the \textit{lhs} when the probabilities admit a LHV model. When instead the probabilities are obtained from quantum mechanics, the supremum of the \textit{lhs} is called the Tsirelson bound \cite{tsirelson80}.
Bell inequalities are often written instead in terms of \textit{correlators} $\mean{A_xB_y}:=p(a{=}b|xy)-p(a{\neq}b|xy)$ when the outcomes $a$ and $b$ can take only two possible values. The prototypical example is the CHSH inequality \cite{chsh1969}
\begin{multline}
\mean{A_0B_0} + \mean{A_0B_1} + \mean{A_1B_0} - \mean{A_1B_1} \le 2,
\end{multline}
which has a Tsirelson bound of $2\sqrt{2}$. For a more in-depth introduction to Bell inequalities see Ref.~\cite{brunner14}.

On the other hand, in the computer science literature it is common to use nonlocal games. Curiously, they were studied for a long time in relation to the complexity class MIP \cite{benor1988} before the connection to nonlocality was noticed \cite{cleve2004}. Again for the bipartite scenario, a nonlocal game is a cooperative game in which a referee sends questions $x,y$ sampled from a probability distribution $\mu(x,y)$ to two parties, Alice and Bob, which then provide answers $a,b$. The referee then accepts their answers with probability $V(a,b,x,y)$, and the parties win the game if the referee has accepted their answers. The local bound of a nonlocal game $G$ is denoted $\omega_\ell(G)$, and it is the maximal probability of winning the game with LHVs, and the Tsirelson bound, denoted $\omega_q(G)$, is the supremum of the probability of winning the game with quantum mechanics. The prototypical example is the CHSH game, introduced by Tsirelson \cite{tsirelson1997}. In it $a,b,x,y \in \{0,1\}$, $\mu(x,y) = 1/4$ and $V(a,b,x,y) = [a\oplus b = xy]$, where $[\cdot]$ are Iverson brackets, i.e., $[\Pi]=1$ if the proposition $\Pi$ is true and $0$ if the proposition $\Pi$ is false. Its local bound is $\omega_\ell(G_\text{CHSH}) = 3/4$ and its Tsirelson bound is $\omega_q(G_\text{CHSH}) = (2+\sqrt2)/4$.

There is a very important sense in which Bell inequalities and nonlocal games are completely equivalent: as remarked in Refs.~\cite{buhrman2010,palazuelos2016}, and as we explore in Section \ref{sec:optimal}, all Bell inequalities can be transformed into a nonlocal game without affecting their ability to detect nonlocality. As we show in Appendix \ref{app:deterministic}, this is not true if we restrict the predicate $V(a,b,x,y)$ to be deterministic, as is often done in the literature. As we show in Appendix \ref{app:deterministiclifting}, however, if we also allow lifting, i.e., embedding the Bell inequality in a scenario with more inputs, then it is possible to transform all two-outcome Bell inequalities into nonlocal games with deterministic predicate. This does not hold for three or more outcomes. This transformation does have a cost, however: as we show in Section \ref{sec:cglmp}, the optimal nonlocal games corresponding to the CGLMP inequalities must have probabilistic predicate, even though they can always be turned into an equivalent nonlocal game with deterministic predicate and smaller gap.

In a wider sense, though, there are relevant differences between these concepts. In fact, there are several advantages to using the nonlocal game formulation:
\begin{enumerate}
\item The local and Tsirelson bounds of a nonlocal game are physically meaningful, and more generally we refer always to the probability of winning the game, as opposed to obtaining some value in the left hand side of the Bell inequality. This is very convenient for statistical analysis, and makes comparison between different nonlocal games meaningful.
\item It immediately suggests a simple and powerful statistic to analyse the results of multiple rounds of playing: the number of victories. The Bell inequality formulation suggests, on the other hand, that we should estimate each individual term and sum these estimates. This not only makes it impractical to do experiments with Bell inequalities that have a large number of terms (which is the rule, not the exception), but this statistic is vulnerable to the memory loophole: as shown in Ref.~\cite{barrett2002}, it becomes possible for the LHVs to use knowledge of the past settings to slightly increase the value of the estimate. As we show in Section \ref{sec:pvalue}, the number of victories statistic is not vulnerable to the memory loophole.
\item It is much more pedagogical. Nonlocal games are easy to understand, and highlight essential features of Bell inequalities: that the questions $x,y$ must be random, that particular answers $a,b$ correspond to particular questions in particular rounds, and that the details of the experimental apparatus are irrelevant.
\end{enumerate}
It's also worth mentioning an advantage of the Bell inequality formulation: they have a direct geometrical interpretation as separating hyperplanes in the set of correlations, and are very convenient to use when searching for the facets of the local polytope, which correspond to the so-called tight Bell inequalities \cite{froissart1981}.

\section{$p$-values}\label{sec:pvalue}

When reporting the result of a Bell experiment, several authors write that they observed a violation of some number of standard deviations above the local bound, where this standard deviation refers to the variance of the experimental statistics \cite{aspect1982,weihs1998,rowe2001}. For the purpose of rejecting LHV models this is not relevant, as discussed in detail in Ref.~\cite{zhang2011}. One relevant figure of merit is the $p$-value, the probability of obtaining a result at least as extreme as the observed data assuming that the null hypothesis is true, in this case that the world is described by LHVs. Indeed the $p$-value was the figure of merit reported in the recent loophole-free Bell tests \cite{hensen2015,giustina2015,shalm2015}.

As shown in Ref.~\cite{elkouss2016}, the $p$-value of obtaining $v$ victories out of $n$ rounds, that is, the probability of obtaining $v$ or more victories out of $n$ rounds with LHVs, is given simply by the binomial distribution\footnote{Note that although the authors restricted the predicate $V(a,b,x,y)$ to be deterministic, their proof holds without change for the general case. The authors treated the general case by demanding the predicate to be deterministic but allowing it to attribute a \emph{score} to the players, instead of just a win or a loss. They could only prove a looser bound in this case, showing that this is a bad choice.}
\begin{equation}
p(G,v,n) := \sum_{k=v}^n \binom{n}{k} \omega_\ell(G)^k (1-\omega_\ell(G))^{n-k},
\end{equation}
even when taking into account the memory loophole. This is a proof that when using the number of victories in the nonlocal game as the statistic the memory loophole allows LHVs to play the game no better than simply playing them independently. The key idea behind this proof is Gill's observation that this statistic is a supermartingale \cite{gill2003}, as was argued informally in Ref.~\cite{barrett2002} and further developed in Refs.~\cite{zhang2011,zhang2013,bierhorst2014}.

When $G$ is played with the optimal quantum strategy, the expected number of victories is $n \omega_q(G)$, and the $p$-value for it is $p(G,\lceil n \omega_q(G) \rceil,n)$, where $\lceil \cdot \rceil$ is the ceiling function. As we show in Appendix \ref{app:pvaluebound},
\begin{equation}
p(G,\lceil n \omega_q(G) \rceil,n) \le (1-\chi_G)^{n\chi_G},
\end{equation}
where $\chi_G := \omega_q(G)-\omega_\ell(G)$ is the gap of the nonlocal game $G$. Note that the upper bound goes down monotonically with increasing gap for any $n$, and moreover it gets arbitrarily close to zero as the gap becomes arbitrarily close to one.

We want to emphasize that the gap $\chi_G$ is a good quantity to maximize if we are interested in a small $p$-value, because in the literature it is common to maximize instead the ratio $\omega_q(G)/\omega_\ell(G)$ (further explored in Section \ref{sec:asymptotics}), but having a large ratio does not imply having a small $p$-value. For example, in the Khot-Vishnoi game (explained in Section \ref{sec:khotvishnoi}) we can get
\begin{equation}
\omega_\ell(G_{\text{KV}_k}) \le \frac{e^2}{k}\mathand \omega_q(G_{\text{KV}_k}) \ge \frac{1}{\log^2(k)},
\end{equation}
for $k \ge 2^3$ and a power of two. It is easy to see that the ratio $\omega_q(G_{\text{KV}_k})/\omega_\ell(G_{\text{KV}_k})$ grows without bound with $k$, but the upper bound we present goes to one.

\section{Optimal nonlocal game for a Bell inequality}\label{sec:optimal}

In this section we derive the optimal nonlocal game corresponding to a Bell inequality, in the sense of maximizing the gap $\chi_G$. To start, let's define things more formally. A behaviour $P$ is a tensor consisting of conditional probabilities for some Bell scenario, $P^{ab}_{xy} := p(ab|xy)$. A Bell functional is defined by a tensor $M$, and the value of the Bell functional on a behaviour $P$ is given by
\begin{equation} \langle M,P \rangle = \sum_{abxy} M^{ab}_{xy} p(ab|xy).\end{equation}
The local bound $L(M)$ is defined as
\begin{equation} L(M) := \max_{P\in\mathcal L} \langle M,P \rangle, \end{equation}
where $\mathcal L$ is the set of local behaviours, and the Tsirelson bound $Q(M)$ is defined as
\begin{equation} Q(M) := \sup_{P\in\mathcal Q} \langle M,P \rangle, \end{equation}
where $\mathcal Q$ is the set of quantum behaviours. A Bell inequality is the expression $\langle M,P \rangle \le L(M)$.

A nonlocal game is defined analogously for a tensor $G$, with the additional restriction that
\begin{equation}\label{eq:Ggame}
G^{ab}_{xy} = \mu(x,y)V(a,b,x,y),
\end{equation}
where $\mu(x,y)$ is the probability that the referee sends questions $x,y$ to the players, and $V(a,b,x,y)$ is the probability that the referee accepts answers $a,b$ on questions $x,y$. If the predicate $V(a,b,x,y)$ is deterministic we call $G$ a deterministic game.

The following theorem derives the conditions for a Bell functional to be a nonlocal game:
\begin{theorem}\label{theo:bell=game}
A Bell functional $M$ is a nonlocal game if and only if it respects positivity
\begin{equation}\label{eq:gamepositivity}
M^{ab}_{xy} \ge 0\qquad\forall a,b,x,y
\end{equation}
and normalisation:
\begin{equation}\label{eq:gamenormalisation}
\sum_{xy} \max_{ab} M^{ab}_{xy} \le 1.
\end{equation}
\end{theorem}
\begin{proof}
To show that the conditions are sufficient, assume that they are satisfied. Then we can define
\begin{equation}
\mu(x,y) := \frac{\max_{ab} M^{ab}_{xy}}{\sum_{x'y'} \max_{a'b'} M^{a'b'}_{x'y'}}
\end{equation}
and
\begin{equation}
V(a,b,x,y) := \frac{1}{\mu(x,y)}{M^{ab}_{xy}},
\end{equation}
where $0/0 = 0$, so that $\mu(x,y)$ will be in fact a probability distribution, as $\mu(x,y) \ge 0$ and $\sum_{xy} \mu(x,y) = 1$, and $V(a,b,x,y)$ a probabilistic predicate, as $0 \le V(a,b,x,y) \le 1$. Furthermore, we have that
\begin{equation}\label{eq:Mgame}
M^{ab}_{xy} = \mu(x,y)V(a,b,x,y),
\end{equation}
so we have sufficiency. To show that the conditions are necessary, assume that \eqref{eq:Mgame} holds. Then 
\begin{align}
\sum_{xy} \max_{ab} M^{ab}_{xy} &= \sum_{xy} \mu(x,y) \max_{ab} V(a,b,x,y) \\ 
				&\le \sum_{xy} \mu(x,y) = 1,
\end{align}
so normalisation is necessary. Since positivity is obvious, we are done.
\end{proof}

A nonlocal game $G$ is equivalent to a Bell functional $M$ if it can be obtained via transformations that have the same effect in any non-signalling behaviour. This is important because such transformations will not change which half-space the inequality $\langle M,P \rangle \le K$ defines on the set of non-signalling behaviours, and in particular if a Bell inequality $\langle M,P \rangle \le L(M)$ is a facet of the local polytope so will be the expression $\langle G,P \rangle \le \omega_\ell(G)$ for the corresponding nonlocal game.

As shown in Ref.~\cite{rosset14}, these transformations are
\begin{enumerate}
\item Adding a constant $c_{xy}$ to each outcome of each setting $x,y$, i.e., 
\begin{equation}\label{eq:translation}
\forall a,b,x,y\quad M^{ab}_{xy}\mapsto M^{ab}_{xy} + c_{xy}.
\end{equation}
This takes $K$ to $K + \sum_{xy}c_{xy}$.
\item Scaling $M$ by a positive constant $d$, i.e.,
\begin{equation}\label{eq:scaling}
M\mapsto d M.
\end{equation}
This takes $K$ to $d K$. Negative constants would require us to switch from upper bounds to lower bounds, so we exclude them for convenience.
\item Adding some no-signalling constraint to $M$. This leaves $K$ invariant.
\end{enumerate}

The next theorem shows the optimal nonlocal game that can be obtained from a Bell functional by considering only the first two transformations:
\begin{theorem}\label{thm:optimalgame}
The nonlocal game with the largest gap $\chi_G$ that can be obtained from a Bell functional $M$ via translation \eqref{eq:translation} and scaling \eqref{eq:scaling} is given by
\begin{equation}
G^{ab}_{xy} := \frac1{\beta+\alpha}\de{M^{ab}_{xy} +\alpha_{xy}},
\end{equation}
Where
\begin{equation}\label{eq:optimalalpha}
\alpha_{xy} := -\min_{a,b} M^{ab}_{xy}; \quad \alpha := \sum_{xy} \alpha_{xy},
\end{equation}
and
\begin{equation}\label{eq:optimalbeta}
\beta_{xy} := \max_{a,b} M^{ab}_{xy}; \quad \beta := \sum_{xy} \beta_{xy}.
\end{equation}
Its gap is given by
\begin{equation}
\chi_G = \frac{Q(M)-L(M)}{\beta+\alpha}
\end{equation}
\end{theorem}
\begin{proof}
First note that the result of the most general translation and scaling of a Bell functional $M$ can be written as
\begin{equation}\label{eq:generaltransformation}
{G^{ab}_{xy}}' := \frac1{\beta+\beta'+\alpha+\alpha'}\de{M^{ab}_{xy} +\alpha_{xy}+\alpha_{xy}'},
\end{equation}
for $\alpha_{xy},\beta_{xy}$ defined as in the statement of the theorem, and some arbitrary constants $\alpha_{xy}'$ and $\beta'$, where again $\alpha' := \sum_{xy} \alpha'_{xy}$.

Moreover, ${G^{ab}_{xy}}'$ is a valid nonlocal game, respecting positivity \eqref{eq:gamepositivity} and normalisation \eqref{eq:gamenormalisation}, if and only if $\alpha_{xy}' \ge 0$ and $\beta' \ge 0$. To see that, note that
\begin{equation}
\forall x,y\quad \min_{ab} \de{M^{ab}_{xy} +\alpha_{xy}+\alpha_{xy}'} = \alpha_{xy}'.
\end{equation}
Since the denominator in equation \eqref{eq:generaltransformation} is required to be positive, we have that ${G^{ab}_{xy}}' \ge 0$ iff $\alpha_{xy}' \ge 0$. Furthermore,
\begin{equation}
\sum_{xy} \max_{ab}{G^{ab}_{xy}}' = \frac1{\beta+\beta'+\alpha+\alpha'}(\beta+\alpha+\alpha'),
\end{equation}
so $\sum_{xy} \max_{ab}{G^{ab}_{xy}}' \le 1$ iff $\beta' \ge 0$.

Both the local and Tsirelson bounds transform as $K \mapsto \frac1{\beta+\beta'+\alpha+\alpha'}\de{K+\alpha+\alpha'}$, so the gap of $G'$ is given by
\begin{equation}
\chi_{G'} = \frac{Q(M)-L(M)}{\beta+\beta'+\alpha+\alpha'},
\end{equation}
and it is maximized by setting $\beta'=\alpha'=0$.
\end{proof}

Note that for this optimal game the \emph{signalling} bound, the highest probability of success attainable with an arbitrary behaviour, is always equal to $1$. This is not always the case for the \emph{non-signalling} bound, the highest probability of success attainable with a non-signalling behaviour.

We now turn our attention to the remaining transformation, adding no-signalling constraints. Since for any $M$ the optimal $\beta$ and $\alpha$ will be given by equations \eqref{eq:optimalalpha} and \eqref{eq:optimalbeta}, our goal is to minimize the sum $\beta+\alpha$ given by these equations over the no-signalling constraints\footnote{In Ref.~\cite{renou2017} the authors also optimized Bell inequalities over the no-signalling constraints, but with the goal of reducing the variance of the quantum statistics, which as we argued before is irrelevant for the purpose of rejecting LHVs.}, in order to maximize the gap $\chi_G$.

The no-signalling constraints are equations of the form
\begin{equation}\label{eq:nsexample}
p(00|00)+p(01|00)-p(00|01)-p(01|01) = 0,
\end{equation}
meaning that the marginal probability that Alice obtains result $0$ does not depend on whether Bob's input is $0$ or $1$. Since these are satisfied by all non-signalling behaviours (by definition), it means that we can add the corresponding coefficients to $M$, scaled by any constant, without changing its effect on non-signalling behaviours. For example, adding equation \eqref{eq:nsexample} times a constant $\gamma$ means transforming $M$ as
\begin{equation} 
M^{ab}_{xy} \mapsto M^{ab}_{xy} + \gamma(\delta_{a0}\delta_{x0}\delta_{y0}-\delta_{a0}\delta_{x0}\delta_{y1}).
\end{equation}
Let then $\{S_i\}_{i=1}^N$ be the set of no-signalling constraints for the corresponding Bell scenario. In the scenario with $s_A,s_B$ inputs and $k_A,k_B$ outputs for Alice and Bob there are $N=(k_A-1)s_A(s_B-1) + (k_B-1)s_B(s_A-1)$ independent ones, where independent means linearly independent and inequivalent under translation \eqref{eq:translation}. They can be chosen as a basis for the signalling vector space defined in Ref.~\cite{rosset14}. Defining then
\begin{equation}
M' := M + \sum_{i=1}^N \gamma_i S_i,
\end{equation}
the problem of maximizing the gap for the nonlocal game obtained from $M$ reduces to solving
\begin{equation}\label{eq:lpoptimalgap}
\min_{\gamma_i} \de{\beta+\alpha} = \min_{\gamma_i} \sum_{xy} \De{\max_{ab} \de{{M^{ab}_{xy}}'} - \min_{ab} \de{{M^{ab}_{xy}}'}}.
\end{equation}
This is a linear programming problem, that can be solved efficiently by numerical methods. We provide an implementation of this algorithm in Python using CVXPY \cite{diamond2016,agrawal2018} as an ancillary file.

We managed to find an analytic solution for two special cases: for the CGLMP inequalities and for when the nonlocal game obtained from $M$ via Theorem \ref{thm:optimalgame} is a unique game, that is, when the predicate is of the form $V(a,b,x,y) = [a=\sigma_{xy}(b)]$ for some permutations $\sigma_{xy}$. Note that XOR games are a particular case of unique games. The solution for the CGLMP inequalities is presented in Section \ref{sec:cglmp}, and for unique games it is to do nothing: when $M$ corresponds to a unique game it is already optimal. In both cases the solution is unique: any change in the amount of no-signalling constraints decreases the gap. We prove this in Appendix \ref{app:uniqueness}.

In both these cases the Bell functional of the optimal solution was orthogonal to the signalling vector space. In general, though, this is not the case, as we found for the $I_{nn22}$ inequalities for $n \ge 3$ (explored below). Moreover, the solution of the linear programming problem $g := \beta+\alpha$ does not in general uniquely determine the local and Tsirelson bounds of the optimal nonlocal game, as they are given by
\begin{equation}
\omega_\ell(G) = \frac1{g}(L(M)+\alpha)\mathand \omega_q(G) = \frac1{g}(Q(M)+\alpha),
\end{equation}
and it is sometimes possible to change $\alpha$ while keeping $g$ constant. We shall see this in the example of the $I_{4422}$ inequality below. 

This raises the question of how to choose $\alpha$ in general. We choose to maximize $\alpha$, as we want the probability of winning the game with quantum mechanics to be as high as possible. We thus need to solve another linear program, maximizing $\alpha$ with the additional constraint that $\beta+\alpha=g$.

This will finally give us optimal, unique, and physically meaningful local and Tsirelson bounds for a given Bell functional. We propose that the bounds so obtained should be taken as \emph{the} local and Tsirelson bounds for the Bell functional.


To illustrate the method and examine some properties of the resulting nonlocal games, we applied it to the CGLMP and $I_{nn22}$ inequalities.

\subsection{CGLMP inequalities}\label{sec:cglmp}

The CGLMP inequalities, introduced in Ref.~\cite{collins2002}, are a family of bipartite Bell inequalities with two inputs per party labelled as $0$ and $1$, and $k \ge 2$ outputs per party labelled from $0$ to $k-1$. They reduce to the CHSH inequality for $k=2$, and are tight for all $k$. Using the form of the inequalities presented in Ref.~\cite{acin2005}, it is easy to see that the transformations in Theorem \ref{thm:optimalgame} take them to the nonlocal game with probabilistic predicate $G_{\text{CGLMP}_k}$ where the probability distribution of the inputs $x,y$ is the uniform one, $\mu(x,y) = 1/4$, and the predicate is
\begin{multline}\label{eq:predicatecglmp}
V(a,b,x,y) = \\ \sum_{i=0}^{k-2} \de{1-\frac{i}{k-1}}\De{a-b = (-1)^{x\oplus y}i + xy \mod k},
\end{multline}
where $[\cdot]$ are Iverson brackets, and $\de{1-\frac{i}{k-1}}$ is the probability that the referee accepts the answers if the corresponding condition is met.

As proven in Appendix \ref{app:uniqueness}, this predicate is the optimal and unique solution of the linear programming problem \eqref{eq:lpoptimalgap}, meaning that it gives the optimal gap for all $k$. Moreover, since the solution is unique, it is not possible to turn $G_{\text{CGLMP}_k}$ into a nonlocal game with deterministic predicate without decreasing the gap. If one accepts the reduction of the gap, it is possible to do so for every $k$. For example, it can be turned into an equivalent game with the deterministic predicate\footnote{To obtain this form, multiply $G_{\text{CGLMP}_k}$ by $4(k-1)$, add to it the no-signalling constraints $p^A(k|x0)-p^A(k|x1)=0$ and $p^B(k|0y)-p^B(k|1y)=0$ with coefficients $\gamma^A(k,x) = (-1)^xk$ and $\gamma^B(k,y) = (-1)^{y+1}k$, and pass it through Theorem \ref{thm:optimalgame}.}
\begin{equation}
V(a,b,x,y) = \De{ (a-b)(-1)^{x\oplus y} \ge xy },
\end{equation}
and the same probability distribution over the inputs, $\mu(x,y) = 1/4$. The gap gets multiplied by $(k-1)/k$, however. This game is equivalent to the simplified CGLMP inequality found in Ref.~\cite{zohren2008}.

Going back to the optimal game, its local bound is
\begin{equation}
\omega_\ell(G_{\text{CGLMP}_k}) = 3/4
\end{equation}
for all $k$. Its Tsirelson bound is not known exactly, but quantum states and measurements are known up to $k=8$ that match the upper bounds given by the second level of the NPA hierarchy within numerical precision \cite{acin2002,navascues2007}. They are summarized in the following table:
\begin{center}
\begin{tabular}{ c | c | c | c}
  $k$ & $\omega_\ell(G_{\text{CGLMP}_k})$ & $\omega_q(G_{\text{CGLMP}_k})$ & $\chi_{G_{\text{CGLMP}_k}}$ \\ \hline
  2 & 0.7500 & 0.8536 & 0.1036\\
  3 & 0.7500 & 0.8644 & 0.1144 \\
  4 & 0.7500 & 0.8716 & 0.1216 \\
  5 & 0.7500 & 0.8770 & 0.1270 \\
  6 & 0.7500 & 0.8812 & 0.1312 \\
  7 & 0.7500 & 0.8847 & 0.1347 \\
  8 & 0.7500 & 0.8877 & 0.1377 \\
\end{tabular}
\end{center}

It has been shown that $\lim_{k \to \infty} \omega_q(G_{\text{CGLMP}_k}) = 1$ \cite{zohren2010}.

\subsection{$I_{nn22}$ inequalities}\label{sec:inn22}

The $I_{nn22}$ inequalities, introduced in Ref.~\cite{collins2004}, are a family of bipartite Bell inequalities with $n\ge 2$ inputs and $2$ outputs per party. They reduce to the CHSH inequality for $n=2$, and are tight for all $n$ \cite{avis2007}. Their Tsirelson bounds are in general not known, and for the $I_{3322}$ inequality might even require infinite-dimensional quantum systems to be reached \cite{pal2010}.

We ran our algorithm to generate an optimal nonlocal game $G_{nn22}$ up to $n=8$. The Tsirelson bounds were upperbounded by the level 2 of the NPA hierarchy. The results are summarized in the following table:
\begin{center}
\begin{tabular}{ c | c | c | c}
  $n$ & $\omega_\ell(G_{nn22})$ & $\omega_q(G_{nn22})$ & $\chi_{G_{nn22}}$ \\ \hline
  2 & 0.7500 & 0.8536 & 0.1036\\
  3 & 0.8000 & 0.8502 & 0.0502 \\
  4 & 0.8333 & 0.8654 & 0.0321 \\
  5 & 0.8571 & 0.8795 & 0.0224 \\
  6 & 0.8750 & 0.8917 & 0.0167 \\
  7 & 0.8889 & 0.9020 & 0.0131 \\
  8 & 0.9000 & 0.9107 & 0.0107 \\
\end{tabular}
\end{center}
Note that for these cases the local bound is equal to 
\begin{equation}
\omega_\ell(G_{nn22}) = \frac{n+1}{n+2}
\end{equation}
within numerical precision.

One interesting feature of the solutions is that for $n\ge 3$ it always required the nonlocal game to have a nonzero projection in the signalling vector space. That is, if one simply takes the unique form from Ref.~\cite{rosset14}, that has zero projection, and calculates the optimal game according to Theorem \ref{thm:optimalgame}, one gets a gap that is smaller than the optimal one. For example\footnote{The Bell functional with zero projection for $n=3$ is available in \href{http://faacets.com/db/canonical/4}{http://faacets.com/db/canonical/4}}, for $n=3$ the gap so obtained is $0.0471$.

Another interesting feature is that for $n \ge 4$ solving the linear program for the optimal gap did not uniquely determine the local and Tsirelson bounds. For example, for $n=4$ it was possible to reduce the local bound up to $0.7778$ without changing the gap. One consequence of this fact is that the nonlocal game with optimal gap does not always have non-signalling bound equal to 1, in this case it was $0.9444$.

The nonlocal game returned by the linear program did not have a particularly simple form, but for the case $n=3$ we managed to use the no-signalling constraints to take it to a nice form without reducing the gap. The tensor is given by
\begin{equation}
G_{3322} = \frac1{10}\left(\begin{array}{rr|rr|rr}
 0 &  2 &  0 &  1 &  0 &  1 \\ 
 2 &  2 &  1 &  0 &  1 &  0 \\\hline
 0 &  1 &  0 &  2 &  1 &  0 \\
 1 &  0 &  2 &  2 &  0 &  1 \\\hline
 0 &  1 &  1 &  0 &  0 &  0 \\ 
 1 &  0 &  0 &  1 &  0 &  0
\end{array}\right),
\end{equation}
where an element $G^{ab}_{xy}$ is written as the element $(a,b)$ of the $2\times2$ submatrix at coordinate $(x,y)$. 

It represents the following nonlocal game: if Alice and Bob are given inputs $(0,1)$, $(1,0)$, $(0,2)$, or $(2,0)$, which happens with probability $1/10$ each, they have to give different answers. If they are given instead inputs $(1,2)$ or $(2,1)$, which again happens with probability $1/10$, they have to give equal answers. The last case is when they get inputs $(0,0)$ or $(1,1)$, which happens with probability $1/5$ each. In this case they have to answer anything other than $(0,0)$. This can be written as a Bell inequality
\begin{multline}
\frac1{10}[ p(a{\neq}b|01) + p(a{\neq}b|10) +p(a{\neq}b|02) \\ +p(a{\neq}b|20) +p(a{=}b|12)+p(a{=}b|21) \\ + 2\de{p(\neg 00|00)+p(\neg00|11)} ] \le \frac45
\end{multline}

\subsection{Diviánszky-Bene-Vértesi inequality}

It is also interesting to consider the Bell inequality introduced by Diviánszky, Bene, and Vértesi in Ref.~\cite{divianszky2017}, that gives the best known lower bound on the real Grothendieck constant of order three, $K_\mathbb{R}(3) \ge 1.4359$. It is a bipartite Bell inequality with 90 inputs and 2 outputs per party, with local bound $L(M_\text{DBV}) = 324\,230\,014$ and Tsirelson bound $Q(M_\text{DBV}) \ge 465\,590\,111$. It is a \emph{full-correlation} Bell inequality, corresponding to a XOR game.

Since XOR games are a particular case of unique games, Appendix \ref{app:uniqueness} implies that the optimal nonlocal game corresponding to it is simply the one given by Theorem \ref{thm:optimalgame}. It has local bound
\begin{equation}
\omega_\ell(G_\text{DBV}) = \frac{718334945}{1112439876} \approx 0.6457
\end{equation}
and Tsirelson bound
\begin{equation}
\omega_q(G_\text{DBV}) \ge \frac{789014993}{1112439876} \approx 0.7093,
\end{equation}
so although the ratio $Q(M_\text{DBV})/L(M_\text{DBV})$ is close to $K_\mathbb{R}(3)$, the maximum possible for a maximally entangled state in dimension 2, the gap is even smaller than the one of the CHSH game.

\section{Violating a Bell inequality with a single shot}\label{sec:singleshot}

If you have played a single round of a nonlocal game $G$ and won, the $p$-value of that victory is simply the local bound of the game $\omega_\ell(G)$. The probability of obtaining this victory when playing with the optimal quantum strategy is its Tsirelson bound $\omega_q(G)$. Since
\begin{equation}
\omega_\ell(G) \le 1 - \chi_G\quad\text{and}\quad\omega_q(G) \ge \chi_G,
\end{equation}
it is enough to construct a family of nonlocal games with gap $\chi_G$ arbitrarily close to one in order to get a single-shot rejection of LHVs for any chosen $p$-value. We describe here two ways of obtaining such a gap: via parallel repetition and via the Khot-Vishnoi game.

\subsection{Parallel repetition}\label{sec:parallelrepetition}

An obvious thing to try is to play $n$ instances of a nonlocal game in parallel, as a single round, instead of playing them in $n$ consecutive rounds as it is usually done. We would expect that, analogous to the consecutive case, the probability of winning with LHVs a fraction of games $\delta$ more than the fraction expected from its local bound decays exponentially with $n$, and the probability of winning such a fraction quantumly goes exponentially to one. 

It is not that simple, though, because in this scenario the LHV model has access to all inputs of each party simultaneously, and this does make it more powerful in general. As shown in \cite{barrett2002}, the probability of winning two parallel CHSH games with LHVs is 10/16, strictly higher than the $\left(3/4\right)^2=9/16$ one gets by playing them independently. The problem can get even more extreme: for the Fortnow-Feige-Lovász game the probability of winning two parallel instances is the same as the probability of winning a single instance, 2/3 \cite{fortnow1989,feige1992}, although for three instances the probability does decrease to 14/27.

Nevertheless, the idea still works, because the probability of winning such a fraction of games still goes down exponentially with $n$, as proven by Rao's concentration bound \cite{rao2011}. To state it, let us define the parallel game more formally. Consider you have a nonlocal game $G$ with $s$ inputs and $k$ outputs per party, who win with probability\footnote{Although Rao stated his theorem only for games with deterministic predicate, it applies without change to games with probabilistic predicate.} $V(a,b,x,y)$ if they give outputs $a,b$ for inputs $x,y$. Its parallel version is then the game $G^n_\delta$ where they play $n$ copies of $G$ in parallel, and they win at $G^n_\delta$ if they win $\lceil n(\omega_\ell(G)+\delta)\rceil$ or more instances of $G$. The concentration bound is then\footnote{This expression is obtained by using an intermediate expression in Rao's proof and setting the constant $25\gamma$ to be $2\delta/3$ instead of $\delta/4$. We have observed numerically that this choice leads to a tighter bound.}
\begin{equation}\label{eq:concentrationbound}
\omega_\ell(G^n_\delta) \le 2 \exp\de{-n t \log\de{\frac{\omega_\ell(G) + \delta - t}{\omega_\ell(G) + 2\delta/3}}}
\end{equation}
where
\begin{equation}
t := \frac{4\delta^2}{4\delta^2+75^2\lceil \log_2k \rceil +75^2\log_2\De{1/(\omega_\ell(G) + 2\delta/3)}}.
\end{equation}

While computing the Tsirelson bound of $G^n_\delta$ is also difficult, we can obtain a good enough lower bound by playing each instance of $G$ independently with the optimal quantum strategy. The probability of winning $\lceil n(\omega_\ell(G)+\delta)\rceil$ or more instances of $G$ for $\delta < \omega_q(G)-\omega_\ell(G)$ is then lowerbounded by the Chernoff bound
\begin{equation}\label{eq:quantumchernoff}
\omega_q(G^n_\delta) \ge 1 - \exp\de{-n D(\omega_\ell(G)+\delta||\omega_q(G)},
\end{equation}
which does go exponentially to one, as we want.

One might consider the possibility of simplifying the discussion by considering nonlocal games for which the Tsirelson bound is one, known as quantum pseudo-telepathy games \cite{brassard2005}, of which a good example is the magic square game \cite{cabello2001,aravind2002}. As noticed in Ref.\,\cite{kempe2007}, they give us an easy way to construct a nonlocal game with gap arbitrarily close to one: we simply demand the players to win \emph{all} parallel instances, as in the ideal case this happens with probability one using quantum mechanics. The $p$-value of such an event is given by Raz's parallel repetition theorem \cite{raz1998,holenstein2006}, which also gives a tighter bound than Rao's concentration bound. It is, however, completely unrealistic to demand a real experiment to win all parallel instances, as it leaves no room for experimental error. Using instead the concentration bound we get a result that is robust against experimental imperfections, and as an added bonus it applies to \emph{any} nonlocal game, not only pseudo-telepathy ones.

Ironically enough, the nonlocal game for which the concentration bound gave the smallest upper bound we found was in fact a pseudo-telepathy game, consisting of two parallel repetitions of the magic square game. In the case of a single repetition, the magic square game has 3 inputs and 4 outputs per player, local bound 8/9, and requires two singlets to be won with probability 1. For the case of two repetitions we showed that local bound is
\begin{equation}
\omega_\ell(G_\text{MS2}) = \frac{66}{81},
\end{equation}
using the algorithm from section \ref{sec:flavien}. Setting $\delta = 1 - \frac{66}{81}-\frac1{100}$, to allow the players to lose $\frac1{100}$ of the games, we find that for a $p$-value of $10^{-5}$ it is sufficient to play
\begin{equation}
n_\text{MS2} = 32\,654\,296
\end{equation}
parallel copies of $G_\text{MS2}$. For the CHSH game $G_\text{CHSH}$, setting $\delta = \frac{2+\sqrt2}4 - \frac34 - \frac1{100}$, and again for a $p$-value of $10^{-5}$, we find that
\begin{equation}
n_\text{CHSH} = 67\,683\,296
\end{equation}
parallel copies are enough. The probability of winning this many instances quantumly is extremely close to one.

It might seem that it is easier to achieve a single-shot violation with $G_\text{MS2}$ than with $G_\text{CHSH}$, but looking only at the number of repetitions is misleading, as we need 4 singlets to play each instance of $G_\text{MS2}$, but only 1 singlet for each instance of $G_\text{CHSH}$. A more meaningful measure of the experimental effort is the dimension of the quantum system required to achieve the single shot violation, which is
\begin{equation}
d_\text{MS2} = 2^{130\,617\,184}
\end{equation}
and
\begin{equation}
d_\text{CHSH} = 2^{67\,683\,296},
\end{equation}
so $G_\text{CHSH}$ is actually better. We also considered parallel repetitions of $G_{\text{CGLMP}_k}$ and $G_{nn22}$, described in Sections \ref{sec:cglmp} and \ref{sec:inn22}, but they always required larger dimensions.

We do not expect these numbers to be close to the true dimension required for a single-shot violation, because the concentration bound is extremely loose. For instance, for small $n$ it gives us a bound very close to 2. 

To have a better idea on what the minimal required dimension is, we investigated the actual local bounds for $G_{\text{CHSH},\delta}^n$ and $G_{\text{MS},\delta}^n$ for the same $\delta$ as before (which requires the players anyway to win all parallel instances for up to 6 instances of the CHSH game and 8 instances of the magic square game). It is known that $\omega_\ell(G_{\text{CHSH},\delta}^1) = 3/4$, $\omega_\ell(G_{\text{CHSH},\delta}^2)=10/16$, and with the algorithm from section \ref{sec:flavien} we calculated that $\omega_\ell(G_{\text{CHSH},\delta}^3)=31/64$. Moreover, using a classical version of the see-saw algorithm, for which we provide an implementation in C as an ancillary file, the best lower bounds we could find for $n=4,5,6$ are $\omega_\ell(G_{\text{CHSH},\delta}^4) \ge 100/256$, $\omega_\ell(G_{\text{CHSH},\delta}^5) \ge 310/1024$, and $\omega_\ell(G_{\text{CHSH},\delta}^6) \ge 1000/4096$. These lower bounds are achievable by using trivial combinations of the optimal strategies for $G_{\text{CHSH},\delta}^2$ and $G_{\text{CHSH},\delta}^3$. A similar phenomenon happened for the Magic Square. It is known that $\omega_\ell(G_{\text{MS},\delta}^1) = 8/9$, we could show that $\omega_\ell(G_{\text{MS},\delta}^2) = 66/81$, and for $n=3,4$ the best lower bounds we found were $\omega_\ell(G_{\text{MS},\delta}^3) \ge 528/729$, $\omega_\ell(G_{\text{MS},\delta}^4) \ge 4356/6561$, again achievable by trivially combining the optimal strategies for lower $n$. If indeed no new strategies appear, it would be true that $\omega_\ell(G_{\text{MS}}^n) \le (\sqrt{66}/9)^n$ (where here we are requiring the players to win all instances), and a mere 113 parallel repetitions of $G_{\text{MS}}$ would be enough for a single shot violation with $p$-value $10^{-5}$.

The main results of this section can be summarised as follows:
\begin{result}
For any nonlocal game $G$ with a quantum violation, it is possible to obtain a single-shot violation for any desired $p$-value $p>0$ with any quantum probability of success $q<1$, by constructing the parallel game $G^n_\delta$ with $n,\delta$ chosen so that the concentration bound \eqref{eq:concentrationbound} implies $\omega_\ell(G^n_\delta) \le p$ and the Chernoff bound \eqref{eq:quantumchernoff} implies $\omega_q(G^n_\delta) \ge q$. Moreover, this single-shot violation is robust against experimental imperfections.
\end{result}

\subsection{The Khot-Vishnoi game}\label{sec:khotvishnoi}

Parallel repetition is not the only way of obtaining a nonlocal game with a gap arbitrarily close to one. As already outlined in Ref.~\cite{kempe2007}, it is also possible to do that with the Khot-Vishnoi game. This nonlocal game was introduced in Ref.~\cite{kempe2007}, based on a construction by Khot and Vishnoi \cite{khot2005}. We present here its formulation from Ref.~\cite{buhrman2010}.

The game $G_{\text{KV}_k}$ is defined by a integer $k \ge 2$, restricted to be a power of 2, and a noise parameter $\eta \in [0,1/2]$. It is a bipartite game, with $2^k/k$ inputs and $k$ outputs per party. As shown in Ref.~\cite{buhrman2010}, its local and Tsirelson bounds respect
\begin{equation}
\omega_\ell(G_{\text{KV}_k}) \le k^{-\frac{\eta}{1-\eta}} \mathand \omega_q(G_{\text{KV}_k}) \ge (1-2\eta)^2,
\end{equation}
where the lower bound on the Tsirelson bound is achieved by a quantum system of local dimension $k$.

In Refs.~\cite{buhrman2010,palazuelos2012}, for example, the parameter $\eta$ is chosen to be close to $\frac12-\frac1{2\log(k)}$, in order to get a large ratio $\omega_q(G_{\text{KV}_k})/\omega_\ell(G_{\text{KV}_k})$. This choice results in the bounds shown in Section \ref{sec:pvalue}, which have a very small gap $\omega_q(G_{\text{KV}_k}) - \omega_\ell(G_{\text{KV}_k})$. Optimizing instead for a large gap, we choose
\begin{equation}
\eta := \frac{\log(\log(\sqrt[4]{k}))}{\log(k)},
\end{equation}
valid for $k \ge 2^6$, which is the same minimal $k$ for which the Khot-Vishnoi game has a quantum advantage.

With this choice of $\eta$ we get
\begin{equation}
\omega_\ell(G_{\text{KV}_k}) \le \frac1{\log(\sqrt[4]{k})}
\end{equation}
and
\begin{equation}
\omega_q(G_{\text{KV}_k}) \ge 1 - \frac{\log(\log(\sqrt[4]{k}))}{\log(\sqrt[4]{k})},
\end{equation}
so the gap gets arbitrarily close to one. It then follows that
\begin{equation}
k_\text{KV} = 2^{577\,079}
\end{equation}
is enough to achieve a $p$-value of $10^{-5}$. For this $k$ we have $\omega_q(G_{\text{KV}_k}) \ge 0.999$, so the single-shot violation is indeed possible.

Note that $k_\text{KV}$ is much smaller than $d_\text{CHSH}$, the smallest Hilbert space dimension for which we could prove a single-shot violation with parallel repetition, but as we argued before we expect this to be only an artefact of the looseness of the concentration bound. Furthermore, the parallel repetition is a much easier experiment to perform than playing $G_{\text{KV}_k}$, as that requires one to do entangled measurements on a gigantic quantum system, whereas parallel repetition requires only independent measurements.

\section{Bounds on the achievable gap}\label{sec:asymptotics}

It is also interesting to consider the maximal gap $\chi_G$ that can be achieved by quantum states of a given dimension. For this purpose it is convenient to introduce the quantity
\begin{equation}
\Xi(G) := \frac1{1-\chi_G},
\end{equation}
which gets arbitrarily large as the gap gets arbitrarily close to one. $\Xi(G)$ is closely related to the so-called \text{largest violation} of a Bell inequality, introduced in Refs.~\cite{junge2010,junge2010b}, which is defined for nonlocal games as
\begin{equation}
\text{LV}(G) := \frac{\omega_q(G)}{\omega_\ell(G)}.
\end{equation}
It is easy to see that 
\begin{equation}
\Xi(G) \le \text{LV}(G),
\end{equation}
and equality holds when $\omega_q(G) = 1$. This implies that upper bounds on $\text{LV}(G)$ are also upper bounds on $\Xi(G)$, and that a nonlocal game with a large $\Xi(G)$ will have all the benefits associated with $\text{LV}(G)$, such as high resistance to noise \cite{junge2010}, in addition to having small expected $p$-value.

We consider then the largest $\Xi$ that can be achieved by a quantum state of local dimension $d$, taking the supremum over all possible nonlocal games, in which case we write $\Xi(d)$. In Appendix \ref{app:mt} we obtain upper bounds on $\Xi(d)$ by extending  the LHV models for noisy quantum states from Ref.~\cite{almeida2007} and the bounds they imply on $\text{LV}(d)$, generalising the technique introduced in Ref.~\cite{palazuelos2012b}. We obtain that
\begin{equation}
\Xi(d) \le \frac{d^2}{1+\left(\frac{d-1}{d} \right)^d (3d-1) } \le \frac{e}3d.
\end{equation}
If we restrict the measurements in the quantum strategies to be projective, we obtain the tighter bound
\begin{equation}
\Xi^\text{proj}(d) \le \frac{d^2}{(d+1)H(d)-d} \le \frac{d}{\log(d) + \gamma -1},
\end{equation}
where $H(d):=\sum_{i=1}^d \frac{1}{i}$ and $\gamma \approx 0.577$ is the Euler–Mascheroni constant.

When $G$ is an XOR game, the results of Tsirelson imply that $\frac{\omega_q(G)-\frac12}{\omega_\ell(G)-\frac12} \le K_\mathbb{R}(2d^2)$, where $K_\mathbb{R}(d)$ is the real Grothendieck constant of order $d$ \cite{tsirelson1987,acin2006}. When in addition the quantum state is restricted to be $\ket{\phi_d}$, the maximally entangled state of local dimension $d$, it holds that $\frac{\omega_q(G)-\frac12}{\omega_\ell(G)-\frac12} \le K_\mathbb{R}(d^2-1)$. Using the fact that $\omega_\ell(G) \le \frac1{\Xi(G)}$, these bounds imply that
	\begin{equation}
	\Xi^\text{XOR}(d) \le \frac{2}{1 + 1/K_\mathbb{R}(2d^2)},
	\end{equation}
and
	\begin{equation}
	\Xi^{\text{XOR},\ket{\phi_d}}(d) \leq \frac{2}{1+1/K_\mathbb{R} (d^2-1)}.
	\end{equation}
	Note that for $d=2$ the assumption that $G$ is an XOR game can be replaced with the assumption that $G$ only has two outcomes per party, or equivalently that the measurements are restricted to be projective. It is known that $1.4359 \le K_\mathbb{R}(3) \leq 1.4644$ \cite{divianszky2017,hirsch2017}, which implies that $\Xi^{\text{XOR},\ket{\phi_2}}(2) \le 1.1885$. Note furthermore that these bounds cannot be tight, because $\omega_\ell(G) = \frac1{\Xi(G)}$ only when $\omega_q(G) = 1$, but for games with two outcomes per party $\omega_q(G) = 1$ implies $\omega_\ell(G)=1$ \cite{cleve2004}.

The results of Sections \ref{sec:parallelrepetition} and \ref{sec:khotvishnoi} give us lower bounds for $\Xi(d)$. The parallel repetition construction give us for any nonlocal game
\begin{multline}
\Xi(G^n_\delta) \ge \Bigg[ \exp\de{-n D(\omega_\ell(G)+\delta||\omega_q(G)} \\ + 2 \exp\de{-n t \log\de{\frac{\omega_\ell(G) + \delta - t}{\omega_\ell(G) + 2\delta/3}}} \Bigg]^{-1},
\end{multline}
which for the CHSH game results in
\begin{equation}
\Xi(d) \ge \frac13 d^\frac{2}{10\,000\,000}.
\end{equation}
The bounds for the Khot-Vishnoi game imply that 
\begin{equation}
\Xi(d) \ge \frac{\log(\sqrt[4]{d})}{\log(\log(\sqrt[4]{d}))+1},
\end{equation}
which is asymptotically smaller.

Both results are very far from the existing upper bound, but we expect it to be actually achievable. As we discussed in Section \ref{sec:parallelrepetition}, we expect a better concentration bound to show that $\Xi(d)$ close to $d^{\frac12\log_2(9/\sqrt{66})} \approx d^{0.07}$ is possible.

\section{Algorithm for calculating local bounds}\label{sec:flavien}

To calculate the local bound of a Bell functional $M$ with $s_A,s_B$ inputs and $k_A,k_B$ outputs for Alice and Bob, the naïve approach is to just try all possible $k_A^{s_A}k_B^{s_B}$ deterministic strategies and take the maximum. Since the number of strategies is exponential in both $s_A$ and $s_B$, this becomes impractical very quickly. A significantly faster algorithm can be obtained if we notice that for any given strategy of Bob it is trivial to determine the corresponding optimal strategy of Alice, so we only need to try all of Bob's strategies. Analogous algorithms based on this observation have been developed in Refs.~\cite{liang2009,brierley2016} for particular cases.

To be more precise, let $P$ the behaviour generated by the deterministic probability distributions $D^A(a|x)$ and $D^B(b|y)$, that is, $P^{ab}_{xy} := D^A(a|x)D^B(b|y)$. The value of the Bell functional with such a behaviour is
\begin{align}
\langle M, P \rangle &= \sum_{abxy} M^{ab}_{xy}D^A(a|x)D^B(b|y) \\
		     &= \sum_{ax} D^A(a|x) \sum_{by}M^{ab}_{xy}D^B(b|y) \\
		     &=: \sum_{ax} D^A(a|x) M_B(a,x)
\end{align}
so to maximize it Alice needs to, for each $x$, output with probability 1 an $a$ that maximizes $M_B(a,x)$. The value of the Bell functional in this case becomes
\begin{equation}
\langle M, P \rangle = \sum_x \max_a M_B(a,x).
\end{equation}
Therefore, to compute the local bound we just need to loop over all $k_B^{s_B}$ strategies for Bob, generating them on the fly to save memory, compute 
\begin{equation}
M_B(a,x) = \sum_{by}G^{ab}_{xy}D^B(b|y) = \sum_{y}G^{af(y)}_{xy},
\end{equation}
since $D^B(b|y) = [b = f(y)]$ for some $f(y)$, compute $\sum_x \max_a M_B(a,x)$, and take the maximum over the optimal value for each of Bob's strategies. Note that this algorithm will be specially good when the nonlocal game is asymmetric, i.e., when $k_A^{s_A} \neq k_B^{s_B}$, in which case we choose Bob as the party with the fewest strategies.

We provide an implementation of this algorithm in C as an ancillary file.

\section{Acknowledgements}
MA has received funding from the European Union's Horizon 2020 research and innovation programme under the Marie Skłodowska‐Curie grant agreement No 801110 and the Austrian Federal Ministry of Education, Science and Research (BMBWF). It reflects only the author's view, the EU Agency is not responsible for any use that may be made of the information it contains. FH acknowledges funding from the Swiss National Fund (SNF) through the early Postdoc.Mobility P2GEP2 181509. MTQ acknowledges funding from the MEXT Quantum Leap Flagship Program (MEXT Q-LEAP) Grant Number JPMXS0118069605, the support of the Austrian Science Fund (FWF) through the SFB project ``BeyondC'', a grant from the Foundational Questions Institute (FQXi) Fund and a grant from the John Templeton Foundation (Project No. 61466) as part of the The Quantum Information Structure of Spacetime (QISS) Project (qiss.fr). The opinions expressed in this publication are those of the authors and do not necessarily reflect the views of the John Templeton Foundation. The authors also acknowledge Princess Bubblegum for performing very demanding calculations.

We dedicate this work to the memory of Boris Tsirelson.

\bibliographystyle{linksen}
\bibliography{biblio}

\appendix

\section{Nonlocal games with probabilistic predicate}\label{app:deterministic}

Here we show that there exists a Bell functional that cannot be transformed into a nonlocal game with deterministic predicate by providing an explicit example. To start with, note that if a Bell functional $M$ is such that for each setting $x,y$ the tensor $M^{ab}_{xy}$ takes only two different values, then it is easy to transform it into deterministic nonlocal game via translation \eqref{eq:translation} and scaling \eqref{eq:scaling}: just take the form from Theorem \ref{thm:optimalgame}. On the other hand, if there are more than two different values in each setting, translation and scaling won't help, as they cannot change the number of different coefficients in a given setting. Therefore, the question of whether a Bell functional can be transformed into a deterministic game reduces to whether we can use the no-signalling constraints to make each setting have only two different values. Intuitively it's clear that this shouldn't be possible: in a scenario with $k$ outcomes per party each setting will have $k^2$ coefficients, but only $2k-2$ no-signalling constraints will act non-trivially on it.

Consider then a Bell functional such that the coefficients of one of the settings are given by
\begin{equation}
R := \begin{pmatrix}
0 & 0 & 0 \\
1 & 2 & 6 \\
6 & 4 & 1
\end{pmatrix}.
\end{equation}
The most general way to transform it with the no-signalling constraints takes it to
\begin{equation}
R' = \begin{pmatrix}
0 & b & b' \\
1+a & 2+a+b & 6+a+b' \\
6+a'& 4+a'+b & 1+a'+b'
\end{pmatrix}.
\end{equation}
To make this setting take only two values, there are three possibilities
\begin{enumerate}
\item All coefficients except $(0,0)$ are equal to each other.
\item At least one of the coefficients $(0,1), (0,2), (1,0)$, or $(2,0)$ are equal to 0.
\item At least one of the coefficients $(1,1), (1,2), (2,1)$, or $(2,2)$ are equal to 0.
\end{enumerate}
We shall examine all these possibilities in turn, and show that all imply that this setting takes at least three different values.
\begin{enumerate}
\item To make all these coefficients equal, we need in particular to make those in the last row equal, which implies that $b=2$ and $b'=5$. But this implies that there are three different coefficients in the first row, so this does not work.
\item To make the coefficient $(0,1)$ equal to zero, we need to set $b=0$. This implies that the coefficients $(1,0)$ and $(1,1)$ are equal to $1+a$ and $2+a$. Since it's not possible to make them equal, at least one of the must be equal to zero. To zero the first one, set $a=-1$. This implies that the coefficient $(1,1)$ is equal to 1, and that the coefficients $(0,2)$ and $(1,2)$ are $b'$ and $5+b'$. Since we already have two different coefficients, these latter two must be equal to either $0$ or $1$. But this is not possible, since their difference is always 5.
The cases $(0,2), (1,0)$, and $(2,0)$ follow from a similar argument, so we omit them.
\item To make the coefficient $(1,1)$ equal to zero, we need to set $b=-2-a$. This implies that the first two columns of $R'$ are $\begin{pmatrix}
0 & -2-a \\
1+a & 0 \\
6+a'& 2+a'-a 
\end{pmatrix}$. To make all these nonzero coefficients equal, we need in particular to set $a=-4$, but this takes the coefficient $(0,1)$ to 2 and the coefficient $(1,0)$ to $-3$, making this submatrix take three different values already, so this does not work. Therefore, at least one of the nonzero coefficients must be equal to zero. This leads us to examine the possibilities of setting $a=2$, $a=-1$, $a'=-6$, and $a'=-2+a$, which all lead to at least three different values in this submatrix.
The cases $(1,2), (2,1)$, or $(2,2)$ follow from a similar argument, so we omit them.
\end{enumerate}
\qed
\section{Lifting nonlocal games with probabilistic predicate}\label{app:deterministiclifting}

Here we show that if we also consider equivalence under liftings, i.e., addition of inputs, it is possible to transform all Bell functionals with two outcomes into nonlocal games with deterministic predicate.

\begin{theorem}
	Let $M$ be a bipartite Bell functional with $s$ inputs and two outputs per party. There exists a deterministic nonlocal game $G$ with $2s$ inputs and two outputs per party which is equivalent to $M$.
\end{theorem}

\begin{proof}
	As remarked in Appendix \ref{app:deterministic}, we only need to make the coefficients $M^{ab}_{xy}$ assume two different values for each pair of inputs $({x},{y})$. Now, for each input $x$ and $y$, let us define an extra input $x'$ and $y'$ and consider an ``extended-input scenario'' where Alice and Bob have now $2s$ inputs each.
	By exploiting the no-signalling constraints
	\begin{multline}
		\alpha_{xy}\big[ p(00| {x} {y}) + p(01| {x} {y})\big]  =\\ \alpha_{xy}\big[p(00| {x} {y'})  + p(01| {x} {y'})\big]
	\end{multline}
	\begin{multline}
		\beta_{xy}\big[ p(00| {x} {y}) + p(10| {x} {y})\big]  =\\ \beta_{xy}\big[p(00| {x'} {y})  + p(10| {x'} {y}) \big],
	\end{multline}
	we can construct a Bell functional $N$ which is equivalent to $M$ via
\begin{gather}
N^{ab}_{xy} := M^{ab}_{xy} + \delta_{a0}\alpha_{xy} + \delta_{b0}\beta_{xy} \\
N^{ab}_{xy'} := -\delta_{a0}\alpha_{xy}\\
N^{ab}_{x'y} := -\delta_{b0}\beta_{xy} 
\end{gather}
	If we set $\alpha_{xy}:=-M^{01}_{xy} + M^{11}_{xy}$ and $\beta_{xy}:=-M^{10}_{xy}+M^{11}_{xy}$, direct calculation shows that for every fixed pair of inputs $({x},{y})$ the coefficients $N^{ab}_{xy}$ can only assume two different values. We can now perform the construction presented in proof of Theorem~\ref{thm:optimalgame} to obtain a deterministic game $G$ which is equivalent to $M$ in this extra-input scenario.
\end{proof}

Note that the example shown in Appendix \ref{app:deterministic} implies that even under liftings it is not possible to transform Bell functionals with three or more outputs per party into a nonlocal game with deterministic predicate.

\section{Bounding the $p$-value}\label{app:pvaluebound}

In this appendix we leave the argument $G$ implicit in order to simplify notation.

\begin{theorem}
For $\omega_\ell \le \omega_q$ it holds that
\begin{equation}
p(\lceil n\omega_q \rceil,n) \le (1-\chi)^{n\chi},
\end{equation}
where $\chi := \omega_q-\omega_\ell$.
\end{theorem}
\begin{proof}
Let $\delta \ge 0$. From the Chernoff bound and the fact that $n (\omega_\ell + \delta) \le \lceil n(\omega_\ell + \delta) \rceil$ we have that
\begin{equation}
p(\lceil n(\omega_\ell + \delta) \rceil,n) \le e^{-nD(\omega_\ell + \delta\|\omega_\ell)}
\end{equation}
where 
\begin{equation}
D(a\|b) := a \log\de{ \frac{a}{b}} + (1-a) \log\de{\frac{1-a}{1-b}}
\end{equation}
is the binary relative entropy. We can lowerbound it by minimizing each term over $\omega_\ell$ individually, resulting in
\begin{equation}
D(\omega_\ell+\delta\|\omega_\ell) \ge \delta \log\de{\frac{1}{1-\delta}},
\end{equation}
and therefore for $\delta = \chi$ we obtain the claimed bound.
\end{proof}

\section{Bounds on the largest violation}\label{app:mt}

	As discussed in the main text, when Alice and Bob share a $d \times d$-dimensional quantum system, the \textit{largest violation} of a nonlocal game $G$ is defined as $\text{LV}(G,d):=\frac{\omega_{q}(G,d)}{\omega_\ell(G)}$, where the quantum value $\omega_{q}(G,d)$ is optimized over all $d\times d$ quantum states and local $d$-dimensional quantum measurements. In this work, we are mainly interested in the particular case where the Bell functional $M$ is a nonlocal game $G$, but the largest violation can be defined for any Bell functional $M$ via $\text{LV}(M,d):=\frac{Q(M,d)}{L(M)}$ where
\begin{equation} L(M) := \max_{P\in\mathcal L} \abs{\langle M,P \rangle}, \end{equation}
where $\mathcal L$ is the set of local behaviours, and 
\begin{equation} Q(M,d) := \sup_{P\in\mathcal Q_d} \abs{\langle M,P \rangle}, \end{equation}
where $\mathcal{Q}_d$ is the set of quantum behaviours with local dimension $d$. We are interested in the largest violation that can be achieved by a  $d \times d$ quantum state in \emph{any} nonlocal game or Bell functional, so we define
\begin{equation}
\text{LV}_G(d) := \sup_G \text{LV}(G,d),
\end{equation}
where the supremum is taken over all nonlocal games, and
\begin{equation}
\text{LV}_M(d) := \sup_M \text{LV}(M,d),
\end{equation}
where the supremum is taken over all Bell functionals. We now provide upper bounds for $\text{LV}_G(d)$ and $\text{LV}_M(d)$ which improve the existing ones \cite{palazuelos2012b,palazuelos2016}.

\begin{theorem}\label{theo:best}
	When Alice and Bob share a $d\times d$-dimensional quantum system it holds that
	\begin{equation}
	\text{LV}_G(d) \le \frac{d^2}{1+\left(\frac{d-1}{d} \right)^d (3d-1) } \le \frac{e}3d
	\end{equation}
and
	\begin{equation}
	\text{LV}_M(d) \le \frac{2d^2}{1+\left(\frac{d-1}{d} \right)^d (3d-1) } -1 \le \frac{2e}3d-1.
	\end{equation}
When the measurements are restricted to be projective it holds that
	\begin{equation}
	\text{LV}^\text{proj}_G(d) \le \frac{d^2}{(d+1)H(d)-d}  \le \frac{d}{\log(d) + \gamma -1},
	\end{equation}
	where $H(d):=\sum_{i=1}^d \frac{1}{i}$ and $\gamma \approx 0.577$ is the Euler–Mascheroni constant, and
	\begin{equation}
	\text{LV}^\text{proj}_M(d) \le \frac{2d^2}{(d+1)H(d)-d}  -1 \le \frac{2d}{\log(d) + \gamma -1}-1.
	\end{equation}
\end{theorem}

	To prove Theorem \ref{theo:best}, we first prove Lemma \ref{lemma:translation}, that shows how the existence of LHV models for a noisy version of a quantum state imply upper bounds on its largest violation. This connection was first made by Palazuelos in Ref.~\cite{palazuelos2012b} and it is generalised here\footnote{More precisely, when compared to our Lemma\,\ref{lemma:translation}, reference \cite{palazuelos2012b} considers only the case of general Bell functional and only the particular case where $\sigma=\frac{I_d}{d^2}$.}. We then show the best LHV models we have, which are the ones from Ref.~\cite{almeida2007} improved by considering a more general kind of noise.

	\begin{lemma}\label{lemma:translation}
		Let $\rho,\sigma$ be quantum states of dimension $d\times d$. If for some $\eta \ge 0$ the state $\rho_\text{LHV}:= \eta \rho + (1-\eta)\sigma$ admits a LHV model, we have $\text{LV}(G,\rho)\leq \frac{1}{\eta}$ for every nonlocal game $G$.

		Moreover, if the state $\sigma$ is separable, the inequality $\text{LV}(M,\rho)\leq \frac{2}{\eta}-1$ holds for every Bell functional $M$.
	\end{lemma}
	\begin{proof}
		For any nonlocal game $G$ and any set of local measurements one can always construct a Bell operator $\mathbf{G}$ such that the winning probability is given by
		\begin{align}
			\tr(\mathbf{G}\rho) &= \omega_q(G,\rho)\\
			&=\sum_{abxy}\mu(x,y)V(a,b,x,y)p_q(ab|xy).
		\end{align}
		Set the Bell operator $\mathbf{G}$ as the one which maximizes the quantum value of $\rho$ in $G$, \ie, $\tr(\mathbf{G}\rho)=\text{LV}(G,\rho)\omega_\ell(G)$. Since $\rho_\text{LHV}$ admits a LHV model, for every nonlocal game $G$ we have that 
\begin{equation}
\tr(\mathbf{G}\rho_\text{LHV})= \eta \tr(\mathbf{G}\rho) + (1-\eta) \tr(\mathbf{G}\sigma) \leq \omega_\ell(G).
\end{equation}
By dividing both sides of this inequality by $\eta\omega_\ell(G)$ and rearranging terms we obtain
	\begin{align}
	\text{LV}(G,\rho)&\leq \frac{1}{\eta} - \frac{(1-\eta)\tr(\mathbf{G}\sigma)}{\eta\omega_\ell(G)}; \\
		  &  \leq\frac{1}{\eta},
	\end{align}
	where the last inequality follows from the fact that $\tr(\mathbf{G}\sigma) \ge 0$.

	For the case where $M$ is a general Bell functional we define $\mathbf{M}$ as the Bell operator such that $\tr(\mathbf{M}\rho)=\text{LV}(M,\rho)L(M)$.

 Since $\rho_\text{LHV}$ admits a LHV model, for every Bell functional $M$ we have that 
\begin{equation}
\tr(\mathbf{M}\rho_\text{LHV})= \eta \tr(\mathbf{M}\rho) + (1-\eta) \tr(\mathbf{M}\sigma) \leq L(M).
\end{equation}
By dividing both sides of this inequality by $\eta L(M)$ and rearranging terms we obtain
	\begin{align}
	\text{LV}(M,\rho)&\leq \frac{1}{\eta} - \frac{(1-\eta)\tr(\mathbf{M}\sigma)}{\eta L(M)}; \\
	                 &\leq \frac{2}{\eta} - 1,
	\end{align}
	where the last inequality follows from the fact that $-\tr(\mathbf{M}\sigma)\leq L(M)$, which holds true because $\sigma$ is separable.
	\end{proof}

	We now extend the results on LHV models first presented in Ref.~\cite{almeida2007}.

	\begin{lemma}\label{lemma:mafalda}
		Let $\rho$ be a quantum state of dimension $d\times d$ and $\rho_B:=\tr_A(\rho)$. Then the state
	\begin{equation}
		\rho_{LHV}:=f \rho + (1-f) \frac{d\, I_d\otimes \rho_B - \rho } {d^2-1}
	\end{equation}
		admits a LHV model when 
\begin{equation}
\frac1f=\frac{d^2}{1+\left(\frac{d-1}{d} \right)^d (3d-1) },
\end{equation}
and for projective measurements it is enough to set 
\begin{equation}
\frac1f=\frac{d^2}{(d+1)H(d)-d}.
\end{equation}

	\end{lemma}

	\begin{proof}
	In Ref.~\cite{almeida2007}, the authors have presented a LHV model for the $d$-dimensional isotropic state. More precisely, let $\ket{\phi_d}:=\frac{1}{\sqrt{d}}\sum_{i=0}^{d-1}\ket{ii}$. For any local dimension $d$, the isotropic state 
	\begin{equation}	
		\phi_d^\eta := \eta \proj{\phi_d} + (1-\eta) \frac{I_d \otimes I_d}{d^2}
	\end{equation}
	admits a LHV model when $\eta=\frac{(3d-1)(d-1)^{d-1}}{(d+1)d^d}$, moreover, if we only consider projective measurements, we can take $\eta=\frac{H(d)-1}{d-1}$.

Simply by re-parametrising the isotropic state in terms of its fidelity with the maximally entangled state $f$ we have 
\begin{equation}
	\phi^\eta_d=f \proj{\phi_d} +(1-f) \left(\frac{I_d\otimes I_d -\proj{\phi_d}}{d^2-1}\right)
\end{equation}
 where $f=\eta+\frac{(1-\eta)}{d^2}.$ With that, we see that the local hidden state model for POVMs of Ref.~\cite{almeida2007} holds for 
\begin{equation}
\frac1f=\frac{d^2}{1+\left(\frac{d-1}{d} \right)^d (3d-1) },
\end{equation}
and for projective measurements for 
\begin{equation}
\frac1f=\frac{d^2}{(d+1)H(d)-d}.
\end{equation}
	The Schmidt decomposition states that, up to local unitary operations, every pure bipartite state can be written as $\ket{\psi}=\sum_{i=0}^{d-1}\lambda_i \ket{ii}$, with $\lambda_i \ge 0$. Hence, every bipartite pure state can be written as \mbox{$\ket{\psi} = \sqrt{d}I_d\otimes F\ket{\phi_d}$} where the operator $F:=\sum_i \lambda_i \ketbra{i}{i}$ may be seem as `local filtering operation'' on Bob's part. Note that $F^2 = \tr_A \ket{\psi}\bra{\psi}$.

	We now claim that the state 
\begin{multline}
	d\left( I\otimes F \right) \phi^\eta_d \left( I\otimes F^\dagger \right)= \\ f \ketbra{\psi}{\psi} + (1-f) \left(\frac{d I_d \otimes F^2-\ketbra{\psi}{\psi}}{d^2-1}\right)
\end{multline}
 has a LHV model. This hold true because the LHV model presented by Almeida \textit{et al} in Ref.~\cite{almeida2007} is also a local hidden state (LHS) model \cite{wiseman2007}. Also, when a local filter is applied in a state with a LHS model, the output state also admits a LHS model \cite{quintino2015}. 

	We finish the proof by pointing out that convex combinations of states with a LHV model also admits a LHV model, hence the state $\rho$ is not required to be pure.
\end{proof}

	We can now prove Theorem\,\ref{theo:best} by combining Lemma\,\ref{lemma:translation} and Lemma\,\ref{lemma:mafalda} and by noting that the state $\frac{d I_d\otimes \rho_B-\rho}{d^2-1}$ is separable for all states $\rho$. To see that, first note that it is enough to consider the case where $\rho$ is pure. Moreover, since local filtering cannot create entanglement, it is enough to consider the case where $\rho$ is the maximally entangled state. Note that the state $\frac{I_d\otimes I_d-\proj{\phi_d}}{d^2-1}$ is ``very close'' to the completely mixed state. Corollary 4 of Ref.~\cite{gurvits2002} formalises this intuition and proves that $\frac{I_d\otimes I_d-\proj{\phi_d}}{d^2-1}$ is separable.

\section{Optimal and unique solutions for the linear programming problem}\label{app:uniqueness}

Here we show the optimal and unique solutions of the linear programming problem \eqref{eq:lpoptimalgap} for unique games and the CGLMP inequalities. The key property of these Bell functionals is that for each individual setting adding no-signalling constraints can only increase the difference between the maximal and minimal coefficients. This implies that the unique and optimal solution is obtaining by adding zero no-signalling constraints.

To prove that, consider that each party has $k$ outcomes, and let $R$ be the $k\times k$ matrix encoding the coefficients of the original Bell functional for some fixed setting $x_0,y_0$. The most general transformation that can be effected on $R$ via the no-signalling constraints is to take it to
\begin{equation}
R' := R + \sum_{i=0}^{k-1} \gamma_i^A S_i^A+\gamma_i^B S_i^B,
\end{equation}
where the element $(a,b)$ of $S_i^A$ and $S_i^B$ is given by
\begin{equation}
S^A_{i,ab} := \delta_{ai}\mathand S^B_{i,ab} := \delta_{ib}.
\end{equation}
Note that these no-signalling constraints are not independent under translation \eqref{eq:translation}, as 
\begin{equation}
\sum_{i=0}^{k-1} S^A_{i} = \mathbf{1}^{(k)} = \sum_{i=0}^{k-1} S^B_{i},
\end{equation}
where $\mathbf{1}^{(k)}$ is the all-ones $k\times k$ matrix. We added this redundancy for symmetry; we shall remove it later by setting $\gamma^A_0 = \gamma^B_0 = 0$.

First we'll show that for unique games and the CGLMP inequalities setting $\gamma^A_i = \gamma^B_i = 0$ for all $i$ indeed minimizes the difference between the maximal and minimal coefficients of $R'$. That is, it is an optimal solution of the linear programming problem
\begin{equation}\label{eq:lprestricted}
\min_{\gamma_i^A,\gamma_i^B} \de{\max(R') - \min(R')}.
\end{equation}
For that, let's consider its dual:
\begin{equation}
\begin{gathered}
   \max_{p,q} \mean{R,p-q} \\
  \text{s.t.}\quad \forall i\ \mean{S^A_i,p-q} = 0,\ \mean{S^B_i,p-q} = 0, \\
\end{gathered}
\end{equation}
where $p,q$ are probability distribution over the $k \times k$ outcomes. Since the primal problem is strictly feasible, from strong duality\footnote{For a more in-depth introduction to convex optimization see Ref.~\cite{boyd2004}.} we know that for all $\gamma_i^A,\gamma_i^B,p,q$ it holds that 
\begin{equation}
\mean{R,p-q} \le \max(R') - \min(R')
\end{equation}
and moreover that 
\begin{equation}
\mean{R,p-q} = \max(R') - \min(R')
\end{equation}
for some optimal ${\gamma^A_i}, {\gamma^B_i}, p, q$. Therefore, to show that $\gamma^A_i = \gamma^B_i = 0$ for all $i$ is indeed an optimal solution we only need to exhibit $p,q$ satisfying the relevant constraints such that 
\begin{equation}\label{eq:strongduality}
\mean{R,p-q} = \max(R)-\min(R).
\end{equation}

For the case where the original Bell functional encodes a unique game, the elements of $R$ are given by
\begin{equation}
R_{ab} := [a=\sigma(b)]\max(R)
\end{equation}
for some permutation $\sigma$ (where $\max(R)$ is the probability that the referee asks questions $x_0,y_0$). We can set then
\begin{equation}
p(a,b) := \frac1k[a=\sigma(b)]
\end{equation}
and
\begin{equation}
q(a,b) := \frac1k[a=\sigma(b+1\!\!\!\mod k)].
\end{equation}
It is easy to see that
\begin{equation}
\forall i\ \mean{S^A_i,p} = \mean{S^A_i,q} = \mean{S^B_i,p} = \mean{S^B_i,q} = \frac1k,
\end{equation}
so the constraints are satisfied, and that equation \eqref{eq:strongduality} holds, so we have optimality.

For the case of the CGLMP inequalities, shown in equation \eqref{eq:predicatecglmp}, $R$ is given by
\begin{equation}
R_{ab} := \sum_{i=0}^{k-1} C_i [a-b = i \mod k]
\end{equation}
for some constants $C_i$. Let then $i_\text{max}, i_\text{min}$ be such that $\max(R) = C_{i_\text{max}}$ and $\min(R) = C_{i_\text{min}}$. An optimal solution of the dual problem will be
\begin{equation}
p(a,b) := \frac1k[a-b=i_\text{max} \mod k]
\end{equation}
and
\begin{equation}
q(a,b) := \frac1k[a-b=i_\text{min} \mod k],
\end{equation}
It is easy to check that they satisfy the constraints and that equation \eqref{eq:strongduality} holds, so as before we have optimality.

To show uniqueness, we need to show that 
\begin{equation}\label{eq:uniqueness}
\max(R') - \min(R') = \max(R)-\min(R)
\end{equation}
implies that $\gamma_i^A = \gamma_i^B = 0$ for all $i$.

For the case of unique games, $R'$ is given by
\begin{equation}
R'_{ab} = \gamma_a^A + \gamma_b^B + [a=\sigma(b)]\max(R),
\end{equation}
and the inequalities $\min(R') \le R'_{ab} \le \max(R')$ imply that for all $b$ we have
\begin{equation}
\gamma_{\sigma(b)}^A + \gamma_b^B \le \max(R') - \max(R)
\end{equation}
and
\begin{equation}
\min(R') \le \gamma_{\sigma(b+1\!\!\!\mod k)}^A + \gamma_b^B.
\end{equation}
Combining these inequalities to eliminate the variables $\gamma_b^B$, we end up with
\begin{equation}
\gamma_{\sigma(b)}^A \le \gamma_{\sigma(b+1\!\!\!\mod k)}^A + \max(R')-\min(R') - \max(R)
\end{equation}
for all $b$, so equation \eqref{eq:uniqueness} implies the chain of inequalities 
\begin{equation}
\gamma_{\sigma(0)}^A \le \gamma_{\sigma(1)}^A \le \ldots \le \gamma_{\sigma(k-1)}^A \le \gamma_{\sigma(0)}^A,
\end{equation}
which implies that $\gamma_i^A = \gamma_0^A = 0$ for all $i$, as claimed. To prove that $\gamma_i^B = \gamma_0^B = 0$ we just need to run the same argument for $b=\sigma^{-1}(a)$.

For the case of the CGLMP inequalities, $R'$ is given by
\begin{equation}
R'_{ab} = \gamma_a^A + \gamma_b^B + \sum_{i=0}^{k-1} C_i [a-b = i \mod k].
\end{equation}
Defining $i_\text{max}$ and $i_\text{min}$ as before, the inequalities $\min(R') \le R'_{ab} \le \max(R')$ imply that for all $b$ we have
\begin{equation}
\gamma_{b+i_\text{max}\!\!\! \mod k}^A + \gamma_b^B \le \max(R') - \max(R)
\end{equation}
and
\begin{equation}
\min(R') - \min(R) \le \gamma_{b+i_\text{min}\!\!\! \mod k}^A + \gamma_b^B.
\end{equation}
Combining these inequalities to eliminate $\gamma_b^B$ as before and assuming that equation \eqref{eq:uniqueness} holds, we end up with
\begin{equation}
\gamma_{b+i_\text{max}\!\!\! \mod k}^A \le \gamma_{b+i_\text{min}\!\!\! \mod k}^A 
\end{equation}
for all $b$. Since these indices are just permutations of $b$ the same argument as before applies and we are done. \qed
\end{document}